\documentclass[11pt]{article}
\usepackage[utf8]{inputenc}
\usepackage[T1]{fontenc}
\usepackage{lmodern}
\usepackage[margin=1in]{geometry}
\usepackage{amsmath}
\usepackage{amsthm}
\usepackage{enumitem}
\usepackage{graphicx}
\usepackage[hidelinks]{hyperref}
\usepackage{microtype}
\usepackage{algorithm}
\usepackage{algorithmicx}
\usepackage[noend]{algpseudocode} 

\newcommand{\Id}[1]{\mathrm{#1}}   
\newcommand{\Fn}[1]{\texttt{#1}}   
\algnewcommand\algorithmicswitch{\textbf{switch}}
\algnewcommand\algorithmiccase{\textbf{case}}
\algdef{SE}[SWITCH]{Switch}{EndSwitch}[1]{\algorithmicswitch\ #1\ \algorithmicdo}{\algorithmicend\ \algorithmicswitch}%
\algdef{SE}[CASE]{Case}{EndCase}[1]{\algorithmiccase\ #1}{\algorithmicend\ \algorithmiccase}%
\algtext*{EndSwitch}%
\algtext*{EndCase}%
\newtheorem{theorem}{Theorem}

\begin{document}

\title{Fully Dynamic Breadth First Search and Spanning Trees in Directed Graphs}
\author{
  Gregory Morse\\
  ELTE Eötvös Loránd University\\
  Budapest, Hungary\\
  \texttt{morse@inf.elte.hu}
  \and
  Tam\'as Kozsik\\
  ELTE Eötvös Loránd University\\
  Budapest, Hungary\\
  \texttt{kto@elte.hu}
}
\date{}

\maketitle

\begin{abstract}
  We study the problem of maintaining a breadth-first spanning tree and
  the induced BFS ordering in a directed graph under edge updates. While
  semi-dynamic algorithms are known, maintaining the spanning tree,
  level information, and numbering together in the fully dynamic setting
  is less developed. This preprint presents a framework for fully
  dynamic BFS in directed graphs together with supporting data
  structures for maintaining the BFS tree, single-source shortest paths,
  and single-source reachability under both insertions and deletions.
\end{abstract}


\section{Breadth-first search (BFS)}\label{sec:bfs}

Breadth-first search (BFS) is a foundational primitive for reachability,
unweighted shortest paths, level decompositions, and traversal-based graph
analysis. In directed graphs, the problem of maintaining a breadth-first
spanning tree (BFST) together with the induced numeric level ordering under
updates is significantly subtler than maintaining distances alone. The
\emph{semi-dynamic} case is classical, but a thorough \emph{fully dynamic}
algorithm that simultaneously preserves the spanning tree, depth information,
and a valid BFS ordering is not presently standard in the literature. Here we
develop a framework for fully dynamic BFS in directed graphs that explicitly
maintains these three objects. As a byproduct, the framework also maintains
single-source shortest paths (SSSP) in the unit-weight setting and is therefore
immediately adaptable to single-source reachability (SSR). It is thus useful to
separate this problem from the more global transitive-closure setting
\cite{DemetrescuI06}, even
though SSR maintenance can itself serve as a building block for maintaining
strongly connected components (SCCs).

We consider a directed graph $G=(V,E)$ with a distinguished source vertex $r$,
where $n=|V|$ and $m=|E|$. Updates consist of edge insertions and deletions.
Unless stated otherwise, the source remains fixed throughout the update
sequence. Our goal is to maintain a BFST rooted at $r$, a level function, and
an explicit global numbering consistent with BFS discovery order. This last
requirement is important: many dynamic shortest-path structures preserve only
distance labels, whereas the actual tree and the associated order are often the
objects required by higher-level algorithms \cite{DemetrescuI06}.

The literature already shows that maintaining reachability alone is highly
nontrivial. Even and Shiloach solved the online edge-deletion problem for SSR,
yielding the classical ES-tree data structure with $\mathcal{O}(1)$ query time
and $\mathcal{O}(mn)$ total update time \cite{10.1145/322234.322235}.
Henzinger, Krinninger, and Nanongkai later gave a sublinear-time randomized
decremental algorithm for single-source reachability and shortest paths in
directed graphs \cite{10.1145/2591796.2591869}. Chechik et al. further obtained
a decremental algorithm that maintains SSR together with SCC information in
$\tilde{O}(m\sqrt{n})$ total update time \cite{7782945}. These results
underscore that maintaining traversal structure in directed graphs is already
challenging even before one insists on preserving an explicit BFST and BFS
numbering.

A BFS tree from a given root node of a digraph is a spanning tree in which each
vertex is first discovered at minimum distance from the root. Such a tree
exists for every finite directed graph and designated root reachable from the
relevant vertices, and it can be constructed by the standard BFS procedure.
Formally, a BFST is a rooted tree $T=(V_T,P)$, where $V_T\subseteq V$ is the set
of reachable vertices and $P$ is the set of parent--child edges. Every edge
$(x,y)\in P$ satisfies $(x,y)\in E$, except that the root is represented by the
convention $\operatorname{pred}(r)=\emptyset$.

Let the distance function from the root be defined recursively by
\[
\operatorname{dist}(x)=
\begin{cases}
0 & x=r,\\
\operatorname{dist}(\operatorname{pred}(x))+1 & x\neq r.
\end{cases}
\]
By construction, BFS guarantees that for every tree edge $(x,y)\in P$, there is
no edge $(x',y)\in E$ with $\operatorname{dist}(x')<\operatorname{dist}(x)$. In
other words, no predecessor at a smaller distance than the chosen parent has an
edge to the same child. This condition does not determine the tree uniquely:
multiple valid BFS trees may exist for the same root and graph.

\paragraph{BFS ordering.}
A BFS ordering is the linear order induced by the maintained BFST and its
per-level tie-breaking rule. Let $\sigma=(v_1,\ldots,v_n)$ be a list of distinct
elements of $V$. We say that $\sigma$ is a BFS ordering if for all
$1\leq i<j<k\leq n$, whenever
$v_i\in \operatorname{pred}(v_k)\setminus \operatorname{pred}(v_j)$, there exists
$q<i$ such that $v_q\in \operatorname{pred}(v_j)$. Intuitively, if a predecessor
of $v_k$ is discovered early, then some predecessor of $v_j$ must have been
discovered even earlier; this captures the layer-respecting nature of BFS.

\paragraph{Complexity of static BFS.}
A properly implemented BFS runs in $\mathcal{O}(m+n)$ time and uses
$\mathcal{O}(n)$ additional space. In the dynamic setting, the space bound
remains linear, but the main difficulty shifts to update time: one seeks to
restructure only the affected region instead of recomputing a complete BFS from
scratch after each change.

The relationship between BFS and depth-first search (DFS) helps clarify why the
fully dynamic problem is delicate. DFS and its associated spanning tree induce
multiple natural orderings (e.g., preorder and postorder), whereas BFS induces
a single level-respecting order. Moreover, a DFS tree in a directed graph can
contain forward edges from a vertex to a proper descendant that is not its tree
child, while a BFS tree cannot exhibit the analogous phenomenon because level
minimality rules it out. Since fully dynamic DFS in directed graphs is already
known \cite{10.14778/3364324.3364329}, it is natural to expect a comparable BFS
result; nevertheless, preserving BFS tree structure and a stable BFS numbering
turns out to require a different set of invariants.

Semi-dynamic BFS in directed graphs is known \cite{FRANCIOSA2001201}. That work
maintains forward-star and backward-star orderings of successor and predecessor
adjacencies and gives asymptotically improved update bounds over naive repeated
recomputation. However, the emphasis there is on the semi-dynamic setting and
on maintaining a BFS tree or shortest-path information rather than on
maintaining a global BFS numbering. Our framework strengthens this by removing
the need for those auxiliary star structures, explicitly preserving a BFS
ordering, and reducing unnecessary scanning in the decremental case.

The earlier semi-dynamic framework can be interpreted as close in spirit to a
fully dynamic algorithm when the update routines are combined carefully, but it
still does not preserve the global BFS ordering needed by some downstream graph
analyses. We extend this line of work by maintaining a numbering that remains
consistent with level order throughout the update sequence. We also examined
lexicographic BFS as a possible refinement, as introduced in the classical
vertex-elimination literature and later developed through partition-refinement
techniques \cite{RoseTL76,HabibMPV00}. While lexicographic BFS imposes a
more rigid tie-breaking discipline and is algorithmically appealing, its
partition-refinement machinery adds considerable complexity, and it is not
clear that this extra structure would improve decremental update bounds. In
particular, whether dynamic lexicographic BFS can provide an alternative to
ES-tree-style reasoning remains open.

Figures~\ref{fig:bfstinit}--\ref{fig:bfstedgedel} illustrate the maintained
objects. In each figure, solid arrows denote BFST edges, dashed arrows denote
non-tree edges of the directed graph, and each vertex label has the form
``vertex [BFS index].'' Figure~\ref{fig:bfstinit} shows an initial graph,
Figure~\ref{fig:bfstedgeadd} shows the effect of inserting an edge that improves
the BFS structure, and Figure~\ref{fig:bfstedgedel} shows the effect of a tree
edge deletion that forces a nontrivial reattachment and renumbering.

\begin{figure}[t]
  \centering
  \includegraphics[width=0.84\linewidth]{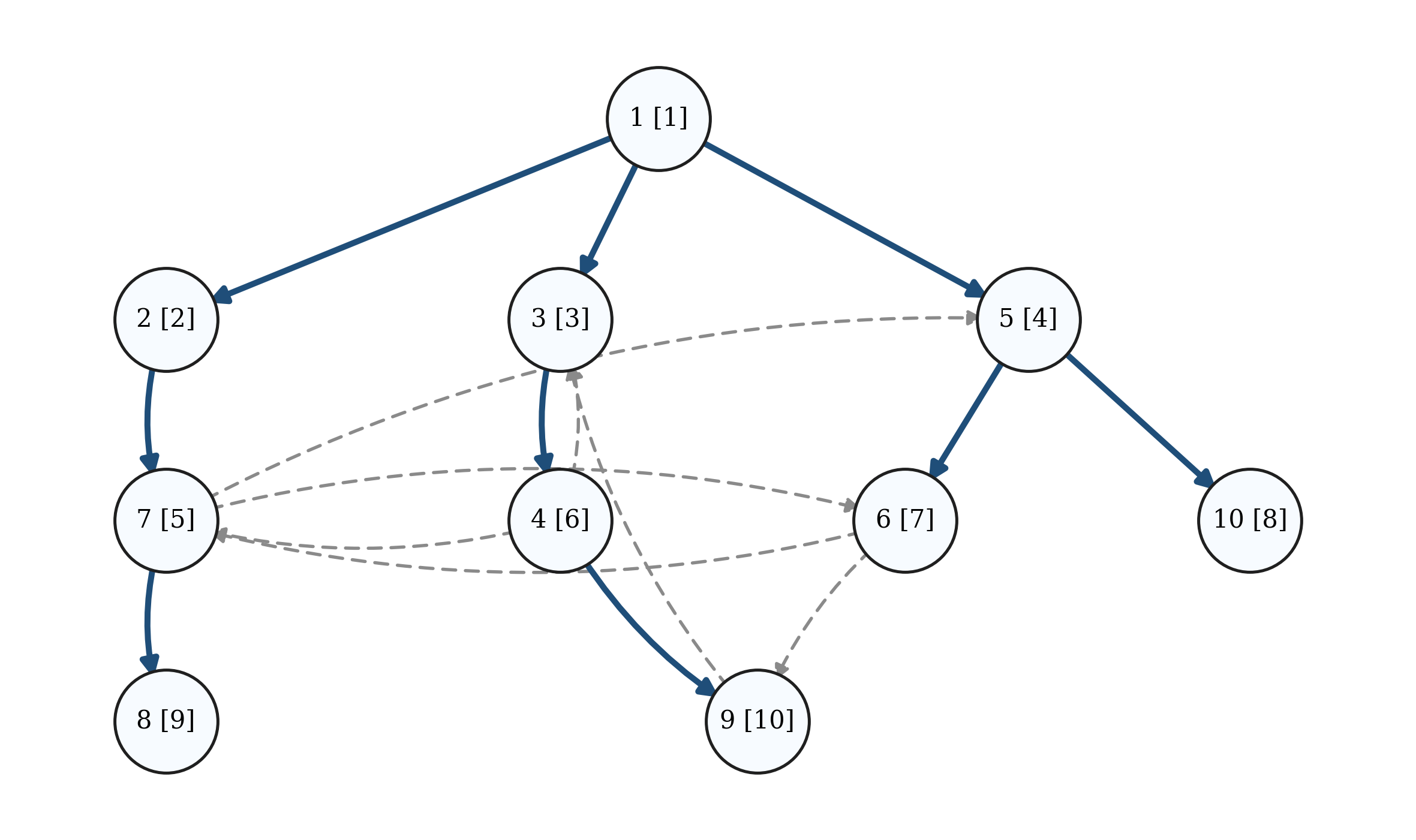}
  \caption{Initial directed graph together with its maintained BFST and global
  BFS numbering. Solid arrows are tree edges; dashed arrows are non-tree graph
  edges.}
  \label{fig:bfstinit}
\end{figure}

\begin{figure}[t]
  \centering
  \includegraphics[width=0.84\linewidth]{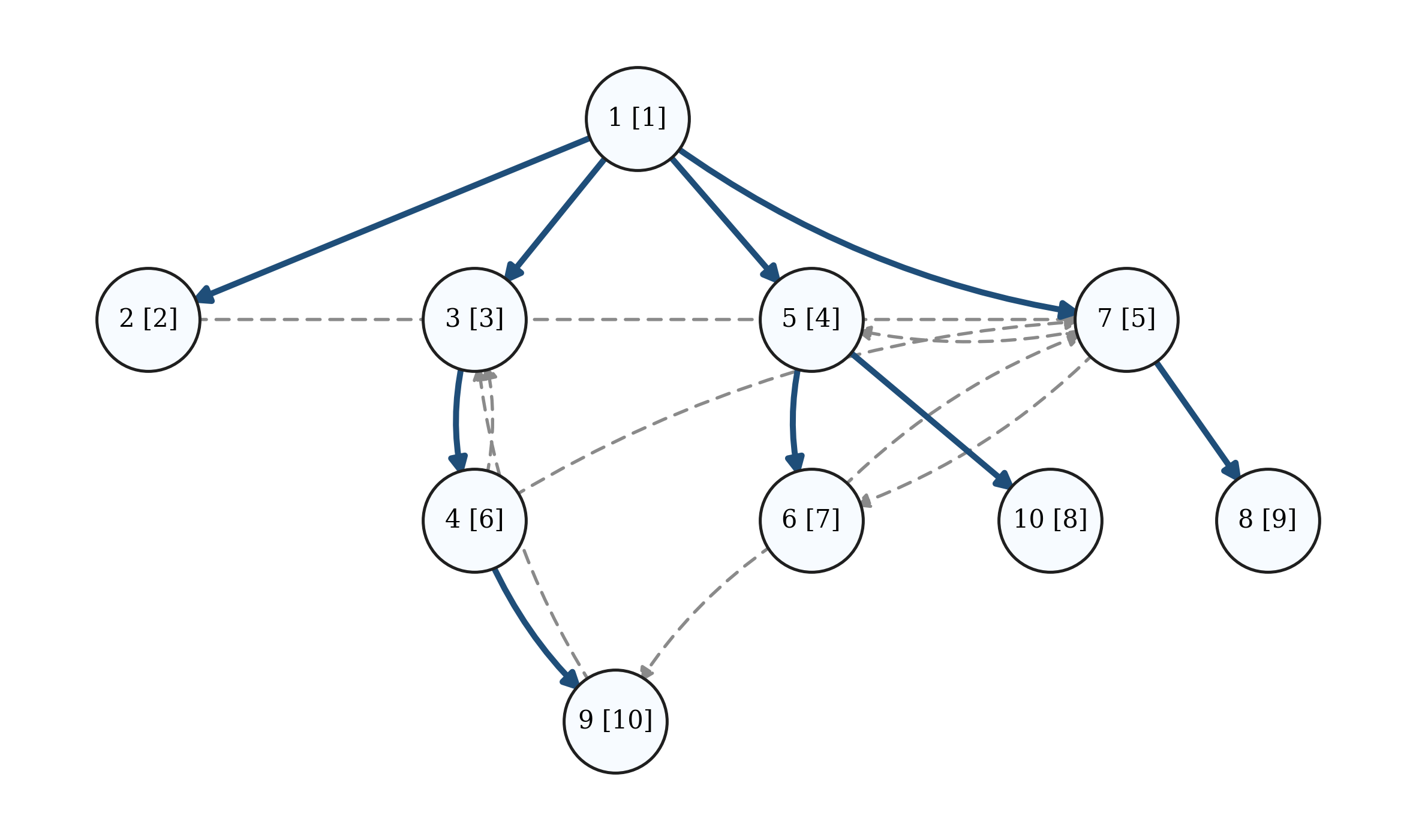}
  \caption{Effect of edge insertion $(1,7)$. The new edge shortens the BFS path
  to vertex $7$, which causes reparenting and renumbering within the affected
  portion of the BFST.}
  \label{fig:bfstedgeadd}
\end{figure}

\begin{figure}[t]
  \centering
  \includegraphics[width=0.74\linewidth]{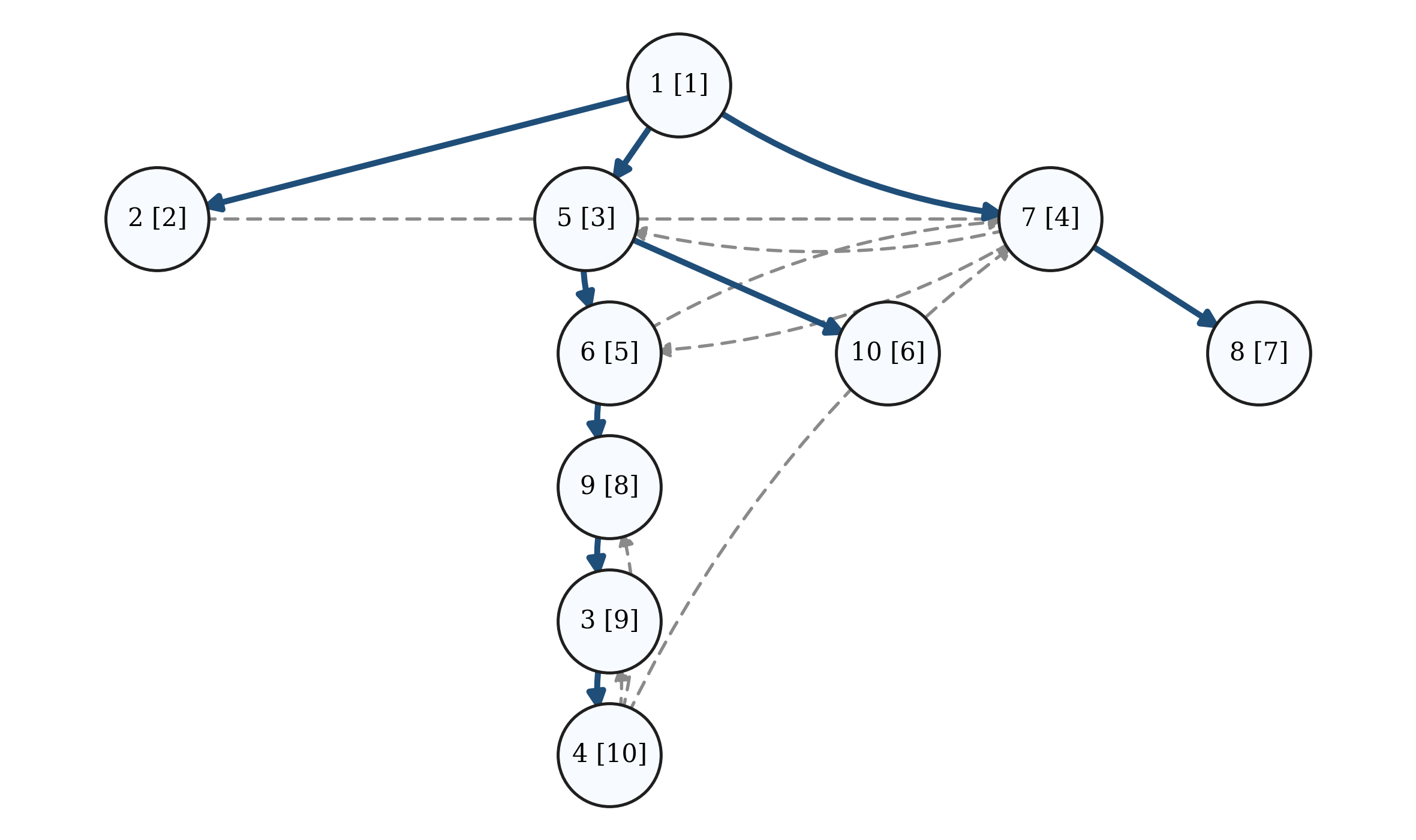}
  \caption{Effect of tree-edge deletion $(1,3)$. The decremental update finds
  replacement parents where possible, raises levels when necessary, and
  preserves a valid BFS numbering on the reachable remainder.}
  \label{fig:bfstedgedel}
\end{figure}

\paragraph{Overview.}
We first fix the BFST representation, numbering scheme, and invariants used
throughout the update process (\S\ref{subsec:bfs-algorithm}). We then present
incremental and decremental maintenance algorithms that preserve the BFST,
depth labels, and BFS order under edge insertions and deletions
(\S\ref{subsec:bfs-incremental}).

\subsection{Algorithm framework}\label{subsec:bfs-algorithm}

The first structural observation is that, much as in dynamic DFS, only certain
updates can alter the maintained tree. On edge insertion, only edges that
create a shorter path to the destination or improve the chosen tie-breaker can
change the BFST, and any such change will convert a non-tree edge into a tree
edge. On edge deletion, only deleted \emph{tree} edges need explicit repair.
From the shortest-path viewpoint, an insertion can never increase the depth of
the destination vertex or any vertex in its affected subtree, whereas a
deletion can never decrease it. Thus insertions may only improve levels, while
deletions may only worsen them or render vertices unreachable.

These monotonicity properties already suggest a localized strategy. Insertion of
an edge $(x,y)$ can alter not only the subtree rooted at $y$ but also any
successor reachable from that subtree whose distance now decreases. Dually,
deletion of a tree edge requires recomputing parent choices from the incoming
edges of vertices in the affected subtree, and the resulting level changes need
not be uniform across that subtree. The essential point is that no unaffected
part of the graph needs to be reconsidered. Since edge scans dominate the cost
of dynamic BFS, minimizing the scanned region is the key design principle.

The decremental case requires additional care because reachability itself may
change. If, after deleting a tree edge, the destination vertex remains
reachable from outside its former subtree, then much of the repair process is
analogous to the incremental case. If not, the algorithm must detect
temporarily disconnected vertices, search for alternative parents in the proper
order, and remove vertices that become unreachable. A two-stage iterative
procedure is therefore preferable to a naive recursive repair: it makes the
reachability cases explicit and avoids repeated rescanning of the same region.

It is also useful to separate the general problem from special application
domains such as control-flow graphs (CFGs). In a CFG-like setting one may wish
to preserve connectivity from a distinguished entry node by means of a
super-source or fictitious edge. That model can simplify reachability but may
force renumbering at the root level whenever new nodes are attached. For such
applications, practical update policies are possible---for instance, by only
allowing a newly inserted node to arrive with a unique predecessor already in
the BFST and no successors---but our framework does not rely on these extra
assumptions.

More generally, one might attempt to root BFS trees at all source SCCs of the
condensation DAG. However, maintaining one BFST per root SCC is worst-case
quadratic, and the resulting trees need not coincide unless those SCCs have the
same reachable structure. Unlike in the DFS setting, a virtual root does not by
itself yield a meaningful single BFST for the entire graph. We therefore adopt
the more general single-source formulation with no additional connectivity
assumptions beyond the chosen root.

Edge insertions can integrate an entire previously unreachable region into the
maintained BFST. Edge deletions, conversely, can disconnect some or all of a
subtree. More subtly, a deleted tree edge may leave every vertex of the former
subtree reachable, but only via replacement parents internal to that subtree;
then the child--parent relation throughout the subtree may change. These are
the genuinely expensive cases, because renumbering must begin at the first
changed vertex and continue through every subsequent emitted vertex in BFS
order.

The maintained state is deliberately compact. The BFST stores parent,
child-successor, and level information, all within $\mathcal{O}(n)$ space. We
also store a global BFS index map $\Id{bfs\_int}$, its inverse array
$\Id{bfs\_revint}$, and a level-boundary array $\Id{bfs\_level}$ that records
the last BFS index of each level. Together these structures allow the update
algorithms to re-emit only the portion of the BFS order that has actually
changed.

Maintaining depth information naively during every tree-edge modification could
cost $\mathcal{O}(n^2)$ in the worst case, because multiple parent changes may
occur during one graph update. Our approach delays those level updates until the
relevant sweep has identified the affected region, after which all necessary
depth changes are applied in one coordinated pass. This reduces the depth
maintenance itself to $\mathcal{O}(n)$ in the size of the changed subtree and is
one of the reasons the fully dynamic update routines remain competitive with
their semi-dynamic predecessors.

\subsection{Maintaining BFS tree, depth, and order under edge updates}\label{subsec:bfs-incremental}
\label{sec:bfs-maintenance}

\paragraph{Why this matters.}
Recomputing a BFST after each update costs $\Theta(m)$ time per change even when
the modification affects only a small local region. Semi-dynamic methods often
maintain distances without preserving a stable per-level order, or they rebuild
the ordering implicitly during each update. Our contribution is a pair of fully
explicit incremental and decremental algorithms that jointly maintain
\begin{enumerate}[label=\textup{(\roman*)}, leftmargin=2.1em]
  \item the BFST parent $\Id{pred}[v]$,
  \item the depth $\Id{level}[v]$,
  \item the tree children $\Id{succ}[u]$,
  \item a global BFS numbering $\Id{bfs\_int}[v]$ and its inverse
  $\Id{bfs\_revint}[i]$, and
  \item per-level range boundaries $\Id{bfs\_level}[\ell]$, the end index of
  level $\ell$ in $\Id{bfs\_revint}$.
\end{enumerate}
This state makes it possible to renumber and relevel only the necessary
vertices, preserve the order of unaffected vertices, and merge a newly
discovered component into the existing BFST without a full recomputation.

\paragraph{Invariants.}
We maintain the following invariants throughout the update sequence:
\begin{enumerate}[label=\textbf{I\arabic*.}, leftmargin=2.3em]
  \item \emph{Depth-parent rule:} for every reachable vertex $v\neq r$,
  $\Id{level}[v]=\Id{level}[\Id{pred}[v]]+1$.
  \item \emph{Contiguous levels:} if $\ell<\ell'$, then all vertices of level
  $\ell$ occupy a contiguous block of indices in $\Id{bfs\_revint}$ that ends at
  $\Id{bfs\_level}[\ell]$, and every vertex of level $\ell'$ appears strictly
  later.
  \item \emph{Stable per-level order:} within a fixed level, the relative order
  is consistent with the BFS tie-breaker induced by the parent’s BFS index. If
  an update does not force a reparenting, the order of unaffected vertices in
  that level is preserved.
\end{enumerate}

\paragraph{Incremental update $(x,y)$ insertion.}
If $x$ is not currently discovered, then inserting $(x,y)$ has no effect on the
maintained BFST. If $y$ is undiscovered but $x$ is discovered, then the update
may merge a previously disconnected region into the existing BFS tree. In the
general case, we first test in constant time whether the new edge is irrelevant:
for example, if $\Id{level}[y]\leq \Id{level}[x]+1$ and the new edge does not
improve the tie-breaker, then no structural change is needed. Otherwise the
edge triggers a layered sweep that begins at the target level
$\ell:=\Id{level}[x]+1$.

The sweep is organized around two small dictionaries. The map $\Id{thisLevel}$
stores, for each current parent candidate, the set of children that must be
promoted into the current level; it is initialized to $\{x:\{y\}\}$. The map
$\Id{nextLevel}$ records promotions that should occur one layer later. We also
maintain a set $\Id{monitor}$ of vertices whose parent need not change but whose
level must be updated, and a set $\Id{exclude}$ that prevents duplicate work.

Within a level, vertices are processed in current BFS order by scanning the
range $[\Id{start}..\Id{end}]$ of $\Id{bfs\_revint}$, while a write pointer
$p=\Id{end}+1$ appends the updated order. For a current vertex $\Id{node}$, we
either fast-forward over a contiguous child block when no changes are pending,
or we inspect its outgoing tree/non-tree edges together with the vertices
explicitly queued in $\Id{thisLevel}[\Id{node}]$. Every promoted or monitored
vertex $q$ receives the new level $\Id{level}[q]=\ell$, is emitted at index $p$
(updating both $\Id{bfs\_int}[q]$ and $\Id{bfs\_revint}[p]$), and increments the
write pointer. Each such $q$ then scans its successors once to test whether a
neighbor $z$ has become closer and should therefore be placed either back into
$\Id{thisLevel}$ (if $z$ also belongs to level $\ell$) or into $\Id{nextLevel}$
(if it must be processed at level $\ell+1$). When the current layer is done, we
set $\Id{bfs\_level}[\ell]=p-1$ and continue only if a next layer has been
created or if newly discovered vertices are still being integrated.

\paragraph{Decremental update $(x,y)$ deletion.}
If $\Id{pred}[y]\neq x$, then deleting $(x,y)$ does not change the BFST.
Otherwise the deleted edge is a tree edge and we first perform a lightweight
pre-scan of the subtree formerly rooted at $y$. This scan serves three
purposes. First, it identifies candidate reattachment points through a map
$\Id{addLevel}[z]$ recording children $q$ that may reattach under predecessor
$z$. Second, it marks a set $\Id{monitor}$ of vertices that remain reachable via
non-tree edges but whose levels may need adjustment. Third, it records a list
$\Id{unreachord}$ of vertices that may become unreachable after the deletion.
Only after this information is collected do we remove $(x,y)$ from the BFST.

Repair then proceeds through the same layered sweep used by the incremental
algorithm, beginning at level $\ell=\Id{level}[x]+1$. For each vertex in the
current layer, we stream its child block in BFS order, re-emitting all visible
successors. If a child $z$ is in $\Id{monitor}$ or in the temporary
$\Id{unreachable}$ set and still claims the current vertex as parent, then it is
releveled to $\ell$ and emitted. If the current vertex belongs to
$\Id{addLevel}$, each candidate child not yet emitted is reparented to it,
assigned the new level, and emitted in order. After the final layer, the
algorithm appends the boundary of the reachable prefix and then scans
$\Id{unreachord}$ backwards, deleting any vertices that remained unreachable,
removing their BFST incidence, and popping their trailing slots from the global
inverse numbering.

\paragraph{Complexity and correctness.}
Let $A$ denote the \emph{affected region}: the vertices whose parent or level
changes, together with the edges incident to those vertices that are actually
examined. The layered sweeps touch each affected vertex and each examined edge
only $O(1)$ times, and all renumbering is carried out as a linear emit into the
$\Id{bfs\_revint}$ stream. Unaffected contiguous child blocks are skipped in
$O(1)$ time via the fast path.

\begin{theorem}[Affected-region work bound]
For a single edge insertion or deletion, the total update time is
\[
O\bigl(|A_V|+|A_E|\bigr),
\]
where $A_V$ and $A_E$ are the affected vertices and examined affected edges,
respectively. In the worst case this is $\Theta(n+m)$, but for localized graph
changes it is linear in the size of the truly impacted subgraph. The same bound
also covers the case in which an insertion merges a previously disconnected
component into the existing BFST.
\end{theorem}

\begin{proof}[Proof sketch]
Invariant \textbf{I1} is preserved whenever a vertex is reparented or confirmed
under its current parent, because the algorithm assigns its level from the
parent’s level in the emitting sweep. Invariant \textbf{I2} follows from the
fact that the sweep writes one level at a time into contiguous positions of
$\Id{bfs\_revint}$ and updates $\Id{bfs\_level}$ only after completing that
level. Invariant \textbf{I3} holds because, within a fixed level, vertices are
processed in existing BFS order and ties are resolved by the parent’s BFS index;
when no reparenting is required, unaffected vertices are emitted in the same
relative order. Every explicit scan or reparenting action can therefore be
charged to a unique affected edge or vertex, while fast-forwarding keeps
unaffected child blocks at constant cost.
\end{proof}

\paragraph{Relation to prior semi-dynamic BFS.}
Unlike semi-dynamic approaches that maintain only distances or that reconstruct
local BFS orderings on demand, our method preserves a consistent per-level order
together with a global numbering. These additional invariants support constant-
time level-boundary maintenance, block skipping for unaffected children, and
clean merging of newly reachable components. This is the principal mechanism
that lets the update routines behave like a true maintained traversal rather
than a sequence of local recomputations.

\paragraph{Full pseudocode.}
To keep the main exposition readable, we place the complete line-by-line
pseudocode in Appendix~\ref{apx:algos}. The appendix contains the precise
incremental and decremental procedures used by our prototype and matches the
high-level descriptions given above.

\appendix
\section{Full pseudocode}\label{apx:algos}

This appendix records the full maintenance procedures for the fully dynamic BFS
framework. The first algorithm handles edge insertion; the second handles the
deletion of a tree edge and the resulting reachability repair.

\subsection{Incremental BFS after edge insertion}
\textbf{What it does.} Reattaches and relabels nodes when adding \(\langle x,y\rangle\)
can shorten the BFS path to \(y\) or its descendants.

\textbf{Inputs/outputs.} Takes adjacency \(\Id{succ}\) and an in-place BFS structure
(levels, parent, numbering). Updates \(\Id{bfs\_tree}\) and per-level spans.

\textbf{Guarantees/complexity.} Touches at most the frontier levels that actually
improve; time is linear in the size of the improved slice. See \S\ref{sec:bfs}.

\begin{algorithm}[H]
\centering
  \caption{Incremental BFS after edge insertion \(\langle x,y\rangle\)}
  \label{alg:inc-bfs-add-edge}
  \begin{algorithmic}[1]
\Require \(\Id{succ}\) (adjacency), \(\Id{bfs\_tree}\) with fields \(\Id{level},\Id{pred},\Id{succ},\Id{virtual\_root},\Id{root\_succ}\);
\(\Id{bfs\_int}\) (numbering), \(\Id{bfs\_revint}\) (inverse numbering), \(\Id{bfs\_level}\) (end-index per level), \(x,y\).
\If{\(x \notin \Id{bfs\_tree.level}\)} \textbf{return} \EndIf
\State \(\Id{newnode} \gets (y \notin \Id{bfs\_tree.level})\)
\If{\(\neg \Id{newnode} \ \land\ \neg\big(\Id{bfs\_tree.level}[y] > \Id{bfs\_tree.level}[x] + 1\ \lor\ 
   (\Id{bfs\_tree.level}[y] = \Id{bfs\_tree.level}[x] + 1 \land \Id{bfs\_int}[\Id{bfs\_tree.pred}[y]] > \Id{bfs\_int}[x])\big)\)}
  \textbf{return}
\EndIf
\State \(\Id{lvl} \gets \Id{bfs\_tree.level}[x] + \big(1 \ \text{if}\ \Id{bfs\_tree.level}[y] \neq \Id{bfs\_tree.level}[x]\ \text{else}\ 0\big)\)
\State \(\Id{start} \gets \begin{cases}
1 & \text{if } x=\Id{bfs\_tree.virtual\_root}\\
\Id{bfs\_level}[\Id{lvl}-2]+1 & \text{otherwise}
\end{cases}\)
\State \(\Id{end} \gets \Id{bfs\_level}[\Id{lvl}-1]\), \quad \(\Id{p} \gets \Id{end}+1\), \quad \(\Id{thisLevel} \gets \{\ x :\ \{y\}\ \}\)
\State \(\Id{nextLevel} \gets \emptyset\),\ \(\Id{exclude} \gets \{y\}\),\ \(\Id{monitor} \gets \emptyset\)
\While{\textbf{true}}
  \State \(\Id{nextstart} \gets \Id{p}\)
  \While{\(\Id{start} \le \Id{end}\)}
    \State \(\Id{node} \gets \Id{bfs\_revint}[\Id{start}]\)
    \State \(\Id{E} \gets (\Id{bfs\_tree.succ}[\Id{node}]\ \text{if}\ \Id{start}\neq 1\ \text{else}\ \Id{bfs\_tree.root\_succ})\); \(\Id{i} \gets |\Id{E}|\)
    \If{\(\Id{node} \notin \Id{thisLevel} \ \land\ (\Id{i}=0 \ \lor\ \Id{bfs\_int}[\Id{E}[0]]=\Id{p})\)} \Comment{fast path: contiguous block}
      \State \(\Id{p} \gets \Id{p} + \Id{i}\); \(\Id{start} \gets \Id{start}+1\); \textbf{continue}
    \EndIf
    \State \(\Id{iter} \gets (\Id{E} \ \text{followed by}\ \Id{thisLevel}[\Id{node}] \ \text{if}\ \Id{node}\in \Id{thisLevel}\ \text{else}\ \Id{E})\)
    \ForAll{\(q \in \Id{iter}\)}
      \State \(\Id{changelvl} \gets (\Id{node}\in \Id{thisLevel} \land q\in \Id{thisLevel}[\Id{node}])\)
      \If{\(q \notin \Id{bfs\_tree.pred}\)} \(\Id{bfs\_tree}.\Fn{add\_node}(q)\) \EndIf
      \If{\(q \notin \Id{monitor} \land \Id{changelvl} \land (\Id{bfs\_tree.pred}[q] \neq \Id{node} \ \lor\ \Id{bfs\_tree.pred}[q] = \Id{bfs\_tree.virtual\_root})\)}
\algstore{incBFS}
\end{algorithmic}
\end{algorithm}
\begin{algorithm}
\centering
\begin{algorithmic} [1]
\algrestore{incBFS}
        \If{\(\Id{bfs\_tree.pred}[q] \neq \Id{bfs\_tree.virtual\_root} \ \lor\ q \in \Id{bfs\_tree.root\_succ}\)}
          \State \(\Id{bfs\_tree}.\Fn{remove\_edge}(\Id{bfs\_tree.pred}[q], q)\)
        \EndIf
        \State \(\Id{bfs\_tree}.\Fn{add\_edge}(\Id{node}, q)\);\quad \(\Id{bfs\_tree.level}[q] \gets \Id{lvl}\)
      \ElsIf{\(q \in \Id{monitor}\)}
        \(\Id{bfs\_tree.level}[q] \gets \Id{lvl}\)
      \EndIf
      \State \(\Id{bfs\_int}[q] \gets \Id{p}\);\quad \(\Id{bfs\_revint}[\Id{p}] \gets q\);\quad \(\Id{p} \gets \Id{p}+1\)
      \If{\(\Id{changelvl} \ \lor\ q\in \Id{monitor}\)}
        \ForAll{\(z \in \Id{succ}[q]\)}
          \If{\(z=q \ \lor\ z \in \Id{exclude}\)} \textbf{continue} \EndIf
          \State \(\Id{reparent} \gets (z \notin \Id{bfs\_tree.level})\ \lor\ (\Id{bfs\_tree.pred}[z]=q)\ \lor\ (\Id{bfs\_tree.level}[z]>\Id{lvl}+1)\ \lor (\Id{bfs\_tree.level}[z]=\Id{lvl}+1 \land (\Id{bfs\_revint}[\Id{bfs\_int}[\Id{bfs\_tree.pred}[z]]] \neq \Id{bfs\_tree.pred}[z]\ \lor\ \Id{bfs\_int}[\Id{bfs\_tree.pred}[z]]>\Id{bfs\_int}[q]\ \lor\ (\Id{node}\in \Id{thisLevel} \land \Id{bfs\_tree.pred}[z]\in \Id{thisLevel}[\Id{node}])))\)
          \If{\(\Id{reparent}\)}
            \If{\(z \in \Id{bfs\_tree.level} \land \Id{bfs\_tree.level}[z]=\Id{lvl}\)}
              \State \(\Id{thisLevel}[q] \gets (\Id{thisLevel}[q] \ \text{or}\ \emptyset)\ \cup\ \{z\}\)
            \Else
              \State \(\Id{nextLevel}[q] \gets (\Id{nextLevel}[q] \ \text{or}\ \emptyset)\)
              \If{\(z \in \Id{bfs\_tree.pred} \land \Id{bfs\_tree.pred}[z]=q\)}
                \State \(\Id{monitor} \gets \Id{monitor} \cup \{z\}\)
              \Else
                \quad \(\Id{nextLevel}[q] \gets \Id{nextLevel}[q] \cup \{z\}\)
              \EndIf
            \EndIf
            \State \(\Id{exclude} \gets \Id{exclude} \cup \{z\}\)
          \EndIf
        \EndFor
      \EndIf
    \EndFor
    \State \(\Id{start} \gets \Id{start}+1\)
  \EndWhile
  \If{\(\Id{lvl} < |\Id{bfs\_level}|\)} \(\Id{bfs\_level}[\Id{lvl}] \gets \Id{p}-1\)
  \ElsIf{\(\Id{newnode} \land \Id{lvl} \ge |\Id{bfs\_level}|-1\)} \(\Id{bfs\_level}.\Fn{append}(\Id{p}-1)\) \EndIf
  \If{\(|\Id{nextLevel}| \neq 0 \ \lor\ (\Id{newnode} \land \Id{p}-1 \neq |\Id{bfs\_int}|)\)}
    \Comment{next layer}
  \ElsIf{\(\Id{p}-1 = |\Id{bfs\_revint}|\)} \(\Id{bfs\_level} \gets \Id{bfs\_level}[\,:\Id{lvl}+1\,]\); \textbf{break}
  \ElsIf{\(\Id{p}-1 = \Id{bfs\_level}[\Id{lvl}-1]\)} \(\Id{bfs\_level} \gets \Id{bfs\_level}[\,:\Id{lvl}\,]\); \textbf{break}
  \EndIf
  \State \(\Id{start} \gets \Id{nextstart}\);\ \(\Id{end} \gets \Id{p}-1\)
  \State \(\Id{thisLevel} \gets \Id{nextLevel}\);\ \(\Id{nextLevel} \gets \emptyset\);\ \(\Id{lvl} \gets \Id{lvl}+1\)
\EndWhile
  \end{algorithmic}
\end{algorithm}

\subsection{Decremental BFS after removing a tree edge}
\textbf{What it does.} Handles the hard case: removing a \emph{tree} edge \(\langle x,y\rangle\)
may disconnect a subtree; the routine rescans for alternative parents or marks
nodes unreachable.

\textbf{Inputs/outputs.} Uses \(\Id{pred}\) and the BFS structure; updates parent/level
and trims numbering for unreachable nodes.

\textbf{Guarantees/complexity.} Work scales with the actually disconnected region,
plus a bounded scan of candidate reattach points. See \S\ref{sec:bfs}.

\begin{algorithm}[H]
\centering
  \caption{Decremental BFS after removing tree edge \(\langle x,y\rangle\)}
  \label{alg:dec-bfs-remove-edge}
  \begin{algorithmic}[1]
    \Require \(\Id{pred}\) (predecessor adjacencies), \(\Id{bfs\_tree}\) with fields \(\Id{level},\Id{pred},\Id{succ},\Id{virtual\_root},\Id{root\_succ}\); \(\Id{bfs\_int}\) (numbering), \(\Id{bfs\_revint}\) (inverse numbering), \(\Id{bfs\_level}\) (end-index per level), \(x,y\).
    \State \textbf{if} \(x \notin \Id{bfs\_tree.level} \ \lor\ y \notin \Id{bfs\_tree.level}\) \textbf{ then return}
    \State \textbf{if} \(\Id{bfs\_tree.pred}[y] \neq x\) \textbf{ then return} \Comment{Only process deleted \emph{tree} edges}
    \State \(\Id{thisLevel} \gets \{\,x:\{y\}\,\};\quad \Id{nextLevel} \gets \emptyset;\quad \Id{addLevel} \gets \emptyset\)
    \State \(\Id{unreachable} \gets \emptyset;\quad \Id{unreachord} \gets [\,]\;;\quad \Id{monitor} \gets \emptyset\)    
    \While{\(|\Id{thisLevel}| \neq 0\)}\Comment{Pre-scan subtree (BFS order) for reattach points}
      \ForAll{\( \Id{node} \in \Id{thisLevel}\)}
        \ForAll{\( q \in \Id{thisLevel}[\Id{node}] \)}
          \ForAll{\( z \in \Id{pred}[q] \)}
            \If{\(z = q \ \lor\ z \notin \Id{bfs\_tree.pred}\)} \State \textbf{continue} \EndIf
            \State \(\Id{addLevel}[z] \gets (\Id{addLevel}[z]\ \text{or}\ \emptyset)\ \cup\ \{\,q\,\}\)
            \If{\(\neg \Id{bfs\_tree}.\Fn{isAncestor}(z, y)\ \lor\ \bigl(z \in \Id{monitor} \land z \notin \Id{unreachable}\bigr)\)}
              \(\Id{monitor} \gets \Id{monitor} \cup \{\,q\,\}\)
            \EndIf
          \EndFor
          \State \(\Id{nextLevel}[q] \gets \Id{bfs\_tree.succ}[q]\)
          \If{\(\Id{node} \notin \Id{monitor} \ \land\ q \notin \Id{monitor}\)}
            \State \(\Id{unreachord}.\Fn{append}(q);\quad \Id{unreachable} \gets \Id{unreachable} \cup \{\,q\,\}\)
          \EndIf
        \EndFor
      \EndFor
      \State \(\Id{thisLevel} \gets \Id{nextLevel};\quad \Id{nextLevel} \gets \emptyset\)
    \EndWhile

    \State \(\Id{lvl} \gets \Id{bfs\_tree.level}[x] + 1\)
    \State \(\Id{start} \gets \begin{cases}
      1 & \text{if } x=\Id{bfs\_tree.virtual\_root}\\
      \Id{bfs\_level}[\Id{bfs\_tree.level}[x]-1] + 1 & \text{otherwise}
    \end{cases}\)
    \State \(\Id{end} \gets \Id{bfs\_level}[\Id{bfs\_tree.level}[x]];\quad \Id{p} \gets \Id{end} + 1;\quad \Id{exclude} \gets \emptyset\)
    \State \(\Id{bfs\_tree}.\Fn{remove\_edge}(x, y)\)

    \While{\textbf{true}}
      \State \(\Id{nextstart} \gets \Id{p}\)
\algstore{decBFS}
\end{algorithmic}
\end{algorithm}
\begin{algorithm}
\centering               
\begin{algorithmic} [1]
\algrestore{decBFS}      
      \While{\(\Id{start} \le \Id{end}\)}
        \State \(\Id{node} \gets \Id{bfs\_revint}[\Id{start}]\)
        \State \(\Id{E} \gets \bigl(\Id{bfs\_tree.succ}[\Id{node}] \ \text{if}\ \Id{start}\neq 1\ \text{else}\ \Id{bfs\_tree.root\_succ}\bigr)\)
        \ForAll{\( z \in \Id{E} \)}
          \If{\(z \in \Id{unreachable} \ \lor\ z \in \Id{monitor}\)}
            \If{\(\Id{bfs\_tree.pred}[z] \neq \Id{node}\)} \textbf{continue} \EndIf
            \State \(\Id{bfs\_tree.level}[z] \gets \Id{lvl};\quad \Id{exclude} \gets \Id{exclude} \cup \{\,z\,\}\)
            \If{\(z \in \Id{unreachable}\)} \(\Id{unreachable} \gets \Id{unreachable} \setminus \{\,z\,\}\) \EndIf
          \EndIf
          \State \(\Id{bfs\_int}[z] \gets \Id{p};\quad \Id{bfs\_revint}[\Id{p}] \gets z;\quad \Id{p} \gets \Id{p}+1\)
        \EndFor

        \If{\(\Id{node} \in \Id{addLevel}\)}
          \ForAll{\( z \in \Id{addLevel}[\Id{node}] \)}
            \If{\(z \in \Id{exclude}\)} \textbf{continue} \EndIf
            \State \(\Id{exclude} \gets \Id{exclude} \cup \{\,z\,\}\)
            \If{\(z \in \Id{unreachable}\)} \(\Id{unreachable} \gets \Id{unreachable} \setminus \{\,z\,\}\) \EndIf
            \If{\(\Id{bfs\_tree.pred}[z] \neq \Id{node}\)}
              \State \(\Id{bfs\_tree}.\Fn{remove\_edge}(\Id{bfs\_tree.pred}[z], z)\)
              \State \(\Id{bfs\_tree}.\Fn{add\_edge}(\Id{node}, z)\)
            \EndIf
            \State \(\Id{bfs\_tree.level}[z] \gets \Id{lvl}\)
            \State \(\Id{bfs\_int}[z] \gets \Id{p};\quad \Id{bfs\_revint}[\Id{p}] \gets z;\quad \Id{p} \gets \Id{p}+1\)
          \EndFor
        \EndIf

        \State \(\Id{start} \gets \Id{start}+1\)
      \EndWhile

      \If{\(\Id{lvl}-1 = |\Id{bfs\_level}|\)} \(\Id{bfs\_level}.\Fn{append}(\Id{end})\)
      \Else \(\Id{bfs\_level}[\Id{lvl}-1] \gets \Id{end}\) \EndIf

      \If{\(\Id{p} > |\Id{bfs\_revint}| \ \lor\ \Id{p} = \Id{start}\)}
        \(\Id{bfs\_level} \gets \Id{bfs\_level}[\,:\Id{lvl}\,]\); \textbf{break}
      \EndIf

      \State \(\Id{start} \gets \Id{nextstart};\quad \Id{end} \gets \Id{p}-1;\quad \Id{lvl} \gets \Id{lvl}+1\)
    \EndWhile

    \If{\(\Id{bfs\_level}[-1] \neq |\Id{bfs\_revint}| - |\Id{unreachable}|\)}
      \State \(\Id{bfs\_level}.\Fn{append}\bigl(|\Id{bfs\_revint}| - |\Id{unreachable}|\bigr)\)
    \EndIf

    \For{\(\Id{idx} \gets |\Id{unreachord}| \ \textbf{down to}\ 1\)}
      \State \(z \gets \Id{unreachord}[\Id{idx}]\)
      \If{\(z \notin \Id{unreachable}\)} \textbf{continue} \EndIf
      \State \(\Id{bfs\_tree}.\Fn{remove\_edge}(\Id{bfs\_tree.pred}[z], z)\)
      \State \(\Id{bfs\_tree}.\Fn{remove\_node}(z)\)
      \State \textbf{delete} \(\Id{bfs\_int}[z]\)
      \State \textbf{delete} \(\Id{bfs\_revint}[\,|\Id{bfs\_revint}|\,]\) \Comment{Pop last inverse-index slot}
    \EndFor
  \end{algorithmic}
\end{algorithm}

\bibliographystyle{plain}
\bibliography{dynamicbfs}

\end{document}